\documentclass[conference]{IEEEtran}
   
\IEEEoverridecommandlockouts
\usepackage{graphicx}
\usepackage{amsmath}
\usepackage[noend]{algpseudocode}
\usepackage{algorithmicx,algorithm}
\usepackage{amsthm}
\usepackage{amsfonts}
\usepackage{subfigure}
\usepackage{color}
\usepackage{cite}
\usepackage{setspace}

\newtheorem{Theorem}{Theorem}

\usepackage{amssymb}
\usepackage{color}
\graphicspath{{images/}}
\usepackage{booktabs}
\usepackage{multirow}
\usepackage{makecell}
\usepackage[numbers,sort&compress]{natbib}

\usepackage{stfloats}
\usepackage{amsmath}
\usepackage{booktabs}
\usepackage{amssymb}
\usepackage{enumitem}
\usepackage{threeparttable}
\usepackage{comment}
\usepackage{bm}
\usepackage{booktabs}      
\usepackage{graphicx}     
\usepackage[table]{xcolor} 
\usepackage{multirow}
\setlist{noitemsep, topsep=2pt, parsep=0pt, partopsep=0pt}
\begin{document}
	
	\title{Thermodynamic Digital Twins: Physics-Native Semantic Synchronization for 6G Kinetic Swarms}
	
\author{Wen-Yu Dong,~\IEEEmembership{Member,~IEEE}, Weiwei Jiang,~\IEEEmembership{Senior Member,~IEEE}, Song Zhao, Rui-Si Han,\\ Qi Bi,~\IEEEmembership{Fellow,~IEEE}, Sheng Chen,~\IEEEmembership{Life Fellow,~IEEE}	%

	\thanks{W.-Y. Dong, S. Zhao and Q. Bi are with Future Technology Research Center, China Telecom Research Institute, Beijing 102209, China (E-mails: dongwy@chinatelecom.cn; zhaosong1@chinatelecom.cn; qibi@chinatelecom.cn)} %
		\thanks{W. Jiang is with the School of Information and Communication Engineering, Beijing University of Posts and Telecommunications, Beijing, 100876, China (Email: jww@bupt.edu.cn)}
	\thanks{R.-S. Han is with Cloud Network Operating System R\&D Center, China Telecom, Beijing 102209, China (E-mail: hanruisi@chinatelecom.cn)} %
\thanks{S. Chen is with the School of Electronics and Computer Science, University of Southampton, Southampton SO17 1BJ, U.K. (E-mail: sqc@ecs.soton.ac.uk).} %
\vspace*{-8mm}
}

\maketitle 
	\pagestyle{empty}
	\thispagestyle{empty}
	
\begin{abstract}
		The synchronization of large-scale kinetic swarms with digital twin networks can generate substantial signaling overhead when predictable motion is repeatedly reported. To reduce such redundancy, we propose the thermodynamic digital twin (TDT), a physics-guided synchronization framework that combines field-based state extrapolation with event-triggered correction. For a closed domain with vanishing boundary flux, the divergence--entropy relation is used as a physical motivation for identifying structural changes in the reconstructed swarm field. Based on Reynolds decomposition, the hierarchical thermodynamic synchronization (HTS) protocol separates low-dimensional macroscopic drift updates from full-state microscopic corrections. A locally measurable smoothed-particle-hydrodynamics descriptor, combining velocity divergence and acceleration, provides an implementation-oriented early-warning trigger, while a transport-error analysis provides a short-interval rule for limiting silent extrapolation. The same extrapolation mechanism is further used as a predictor-like delay compensator in the closed-loop digital twin. Across the evaluated numerical settings, TDT reduces the generated edge-to-cloud synchronization payload relative to the considered baselines and exhibits a favorable empirical trade-off between tracking error and communication cost. The results also indicate improved delay tolerance in the evaluated control loop.
\end{abstract}

	
\section{Introduction}\label{S1}
\IEEEPARstart{L}{arge-scale} kinetic systems span autonomous unmanned aerial vehicle (UAV) networks \cite{DongJSAC,DongTCOM}, satellite--aerial networks \cite{Tian2023}, and massive low-Earth-orbit (LEO) satellite constellations \cite{IoTJDong}. In flying ad hoc networks (FANETs), rapid node mobility induces frequent topology and link variations \cite{Yang2025FANET,Wang2025FANET}, making timely digital-state maintenance challenging. Such systems increasingly rely on 6G digital twin networks (DTNs) for cyber--physical synchronization and closed-loop control. For a swarm of $N$ agents, repeatedly transmitting every agent state incurs an $\mathcal{O}(N)$ payload, creating substantial overhead on bandwidth-limited edge-to-cloud synchronization links.

To mitigate this bottleneck, existing digital-twin research has focused on low-latency networking and edge association \cite{Sun2020}, joint communication--computation optimization \cite{Yang2024}, and data-driven or semantic twin construction \cite{Li2021SemanticDT}. These studies improve update delivery or state reconstruction but do not directly determine whether a model-predictable update is required by the downstream task \cite{Yang2023Semantic}. Related policies include age-of-information (AoI)-aware scheduling \cite{Kaul2012}, event-triggered control \cite{Heemels2012}, value-of-information (VoI) methods \cite{Soleymani2022}, and model-based tracking \cite{Blom1988}. Our design is closest to event-triggered and model-based schemes, but selects among silence, a macro update, and a full-state correction using a locally estimated field-deformation descriptor.
Fluid-based frameworks model non-stationary networks through density, velocity, and flux fields \cite{Dong2026FSTSG}, with extensions to UAV-logistics provisioning \cite{DongDigitalTides2026}, LEO mega-constellation interference \cite{DongFSAT}, and loop-free information transport in ultra-dense networks \cite{Dong2026VDR}; however, they do not address synchronization timing and granularity.

	To reduce redundant synchronization under this semantic premise, we propose the thermodynamic digital twin (TDT), which combines field-based extrapolation with hierarchical event-triggered correction. Rather than claiming the divergence--entropy relation as a new information-theoretic identity, we use its closed-domain form as physical motivation for constructing a locally measurable synchronization cue. The resulting hierarchical thermodynamic synchronization (HTS) protocol combines low-dimensional macroscopic drift updates, smoothed-particle-hydrodynamics (SPH)-based microscopic change detection, and predictor-like state propagation between synchronization events.

The main contributions of this paper are summarized as follows.
\begin{itemize}
\item \textit{Hierarchical synchronization architecture:}
We develop the HTS protocol to separate coherent macroscopic drift from local structural deformation. The protocol selects among silence, a constant-size macro update, and a full-state micro-correction, thereby avoiding repeated full-state transmission during coherent motion.
	
	\item \textit{Locally measurable triggering and silent-interval design:}
		We construct an SPH-based kinematic descriptor that combines local velocity divergence and acceleration as an implementation-oriented early-warning trigger. Separately, an $L_1$ transport-error analysis yields a short-interval design rule for the maximum silent interval, complementing the practical trigger.
	
\item \textit{Predictive synchronization and delay-tolerance evaluation:}
The digital twin propagates the reconstructed field between updates using the continuity equation and applies the same extrapolation mechanism for predictor-like delay compensation. Numerical evaluations examine the divergence--entropy relation, computational scalability, communication--tracking trade-offs, and closed-loop delay sensitivity against periodic, reactive, and Transformer-based baselines.
\end{itemize}

\section{System Model and Problem Formulation}\label{S2}
\begin{figure*}[!b]
	\vspace*{-3mm}
	\centering
	\includegraphics[width=1\textwidth]{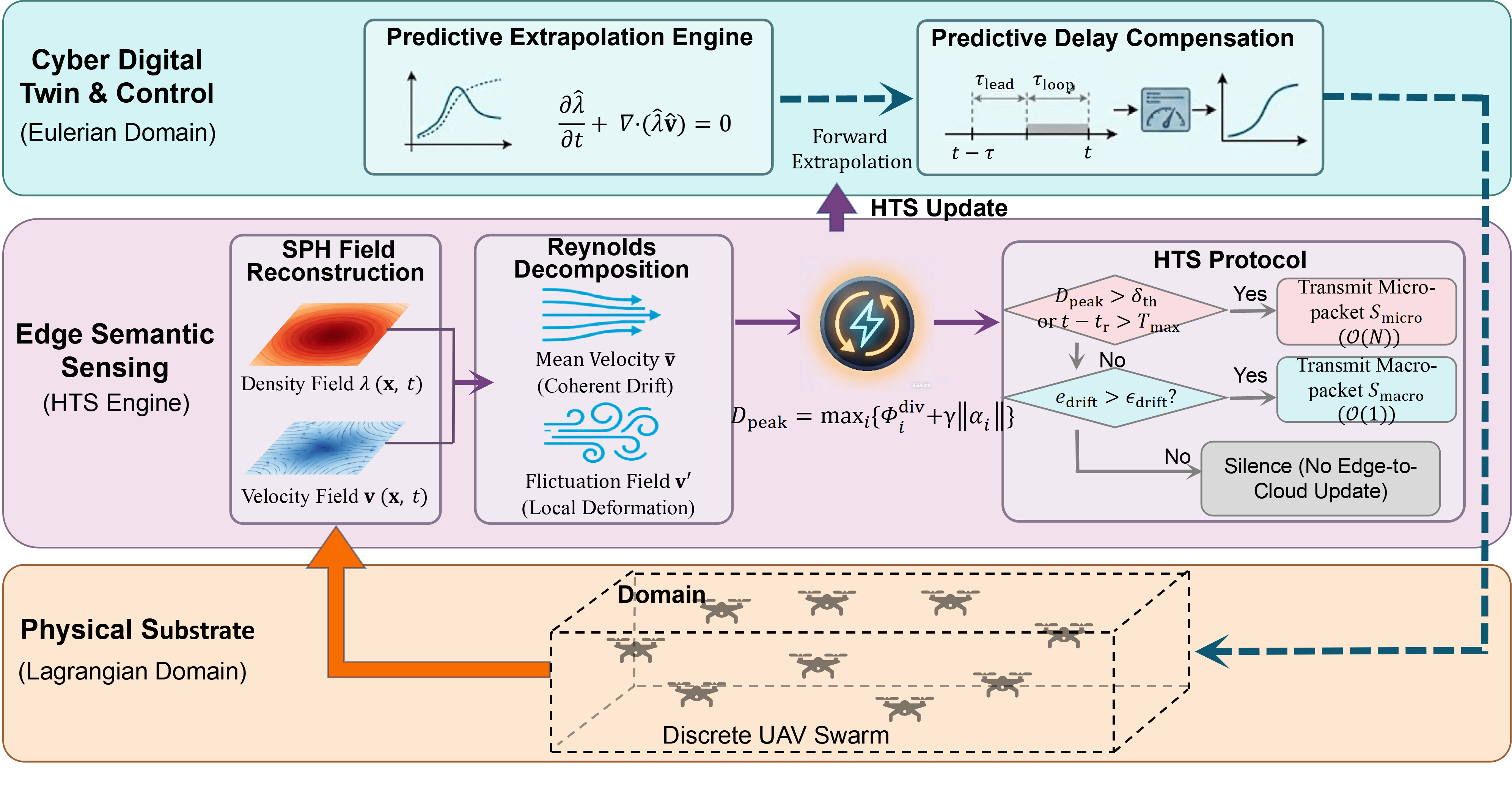}
	\vspace*{-8mm}
	\caption{Architecture of the TDT framework. UAV-to-edge telemetry supports local field reconstruction, while HTS selects silence, constant-size macro updates, or full-state micro corrections on the evaluated edge-to-cloud synchronization link. The cloud extrapolation engine further supports predictor-based delay compensation.}
	\label{fig:system_model} 
\end{figure*}

As illustrated in Fig.~\ref{fig:system_model}, TDT comprises a physical swarm, an edge sensing and synchronization engine, and a cloud extrapolation engine. We consider $N$ agents, indexed by $\mathcal{N}=\{1,\ldots,N\}$, moving within a bounded three-dimensional domain $\Omega\subset\mathbb{R}^3$. We focus on edge-to-cloud synchronization; UAV-to-edge telemetry acquisition is treated as an underlying sensing/access process and is excluded from the reported synchronization payload. To reduce population-dependent modeling complexity, the discrete swarm is represented by continuous Eulerian fields.

\subsection{Physical Substrate Modeling}\label{sec:physical_substrate}

While the physical agents possess discrete states $\{\mathbf{x}_i(t), \mathbf{v}_i(t)\}_{i \in \mathcal{N}}$, we reconstruct a continuous macroscopic density field $\lambda(\mathbf{x}, t)$ using the mesh-free smoothed particle hydrodynamics (SPH) formalism \cite{Price2012_SPH}. The spatial density and the regularized velocity field $\mathbf{v}(\mathbf{x}, t)$ are defined as:
\begin{align} 
	\lambda(\mathbf{x}, t) &= \sum_{i \in \mathcal{N}} W\left( \|\mathbf{x} - \mathbf{x}_i(t)\|, h \right), \label{eq:SPH_density} \\
	\mathbf{v}(\mathbf{x}, t) &= \frac{\sum_{i \in \mathcal{N}} \mathbf{v}_i(t) W(\|\mathbf{x} - \mathbf{x}_i(t)\|, h)}{\lambda(\mathbf{x}, t) + \vartheta}, \label{eq:SPH_velocity}
\end{align}
where $W(\cdot, h)$ is the normalized smoothing kernel with smoothing length $h$, and $\vartheta \to 0^+$ is a regularization term preventing singularities in sparse regions. 

Following the fluid-spatiotemporal field representation of non-stationary network populations in \cite{Dong2026FSTSG}, we characterize the swarm through its macroscopic density and velocity fields. Under the closed-domain and agent-conservation assumptions adopted here, the evolution of $\lambda(\mathbf{x},t)$ satisfies the continuity equation:
\begin{equation} 
	\frac{\partial \lambda}{\partial t} + \nabla \cdot (\lambda \mathbf{v}) = 0.
	\label{eq:physical_continuity}
\end{equation}
Using the material derivative $\frac{D}{Dt}\triangleq\frac{\partial}{\partial t}+\mathbf{v}\cdot\nabla$, \eqref{eq:physical_continuity} gives $\frac{D\lambda}{Dt}=-\lambda\nabla\cdot\mathbf{v}$. Hence, velocity divergence provides a local measure of volumetric compression or expansion. It does not capture divergence-free deformation such as shear or rotation and is therefore used here as a lightweight deformation cue rather than a complete descriptor.

\subsection{Thermodynamic Digital Twin Architecture}\label{sec:tdt_architecture}

The TDT architecture implements a split-computing paradigm where the edge node functions as a semantic extractor and the cloud acts as an extrapolation engine.

\subsubsection{Edge Semantic Sensing}
The edge node continuously aggregates local telemetry to assess macroscopic drift and local structural variation. The mean velocity
$\bar{\mathbf{v}}(t)=|\mathcal{N}|^{-1}\sum_{i\in\mathcal{N}}\mathbf{v}_i(t)$
tracks the global transport drift. Concurrently, to detect local deformation and maneuver changes, the edge computes the kinematic descriptor
$\mathcal{D}_{\text{peak}}(t)$ using a compact-support SPH formulation: 
\begin{equation}
	\mathcal{D}_{\text{peak}}(t)
	=
	\max_{i\in\mathcal N}
	\left\{
	\left|
	\sum_{j\in\mathcal N_i}
	\frac{1}{\lambda_j}
	\mathbf v_{ji}\cdot\nabla_i W_{ij}
	\right|
	+\gamma\|\mathbf a_i\|
	\right\},
	\label{eq:sph_divergence}
\end{equation}
where
$\mathcal{N}_i=\{j:\|\mathbf{x}_i-\mathbf{x}_j\|\leq\kappa h\}$ is the neighbor set,
$W_{ij}=W(\|\mathbf{x}_i-\mathbf{x}_j\|,h)$,
$\lambda_j=\lambda(\mathbf{x}_j,t)$,
$\mathbf{v}_{ji}=\mathbf{v}_j-\mathbf{v}_i$,
and $\mathbf{a}_i=\dot{\mathbf{v}}_i$.
The operator $\nabla_iW_{ij}$ denotes the kernel gradient with respect to $\mathbf{x}_i$, while $\gamma$ balances the divergence and acceleration terms. The descriptor therefore responds to both local compression/expansion and rapid maneuver changes.

\subsubsection{Cloud Extrapolation Engine}
During the silent interval $t \in (t_k, t_{k+1})$ between synchronization events, the DTN autonomously extrapolates the digital density field $\hat{\lambda}(\mathbf{x}, t)$. Exploiting the frozen-flow assumption, the twin utilizes the last synchronized velocity field $\hat{\mathbf{v}}(\mathbf{x}, t_k)$ to advance the state via the integral form of the continuity equation:
\begin{equation} 
	\hat{\lambda}(\mathbf{x}, t) = \hat{\lambda}(\mathbf{x}, t_k) - \int_{t_k}^{t} \nabla \cdot \left( \hat{\lambda}(\mathbf{x}, \tau) \hat{\mathbf{v}}(\mathbf{x}, t_k) \right) \mathrm{d}\tau.
	\label{eq:digital_extrapolation}
\end{equation}
This physics-based extrapolation reduces repeated edge-to-cloud updates during predictable motion.

\subsection{Synchronization Metrics and Problem Formulation}
\label{sec:problem_formulation}

	For massive kinetic swarms, formation coverage and other field-level tasks depend primarily on the spatial distribution of agents rather than their individual identities. We therefore normalize the physical and extrapolated intensity fields as
	\begin{equation}
		p(\mathbf{x},t)
		\triangleq
		\frac{\lambda(\mathbf{x},t)}{N},
		\qquad
		\hat{p}(\mathbf{x},t)
		\triangleq
		\frac{\hat{\lambda}(\mathbf{x},t)}{N}.
		\label{eq:normalized_fields}
	\end{equation}

	We report the Kullback--Leibler (KL) divergence as a supplementary diagnostic of distributional mismatch:
	\begin{equation}
		\mathcal{E}_{\mathrm{KL}}(t)
		\triangleq
		D_{\mathrm{KL}}
		\bigl(
		p(\cdot,t)
		\,\Vert\,
		\hat{p}_{\xi}(\cdot,t)
		\bigr),
		\label{eq:kl_diagnostic}
	\end{equation}
	where $\hat{p}_{\xi}$ is obtained by lower-clipping $\hat{p}$ at a small positive value $\xi$ and renormalizing it over $\Omega$. Because this numerical regularization changes the reference density, $\mathcal{E}_{\mathrm{KL}}$ is used only for diagnostic evaluation. The formal tracking metric is the normalized $L_1$ field error
	\begin{align}
		\mathcal{E}_{1}(t)
		&\triangleq
		\left\|
		p(\cdot,t)-\hat{p}(\cdot,t)
		\right\|_{1}
		\nonumber\\
		&=
		\frac{1}{N}
		\left\|
		\lambda(\cdot,t)-\hat{\lambda}(\cdot,t)
		\right\|_{1},
		\label{eq:l1_tracking_error}
	\end{align}
	which is adopted as the object of the formal extrapolation analysis and synchronization constraint.

	Let
	$\pi=\{(t_k,q_k)\}_{k\in\mathbb N}$
	denote a synchronization policy, where $t_k$ is an update time and
	$q_k\in\{\mathrm{macro},\mathrm{micro}\}$ is the transmitted packet type. If $B_{q_k}$ denotes the corresponding generated synchronization payload, the design objective is
	\begin{equation}
		\begin{aligned}
			\min_{\pi}
			\quad&
			\bar{C}(\pi)
			=
			\limsup_{T\rightarrow\infty}
			\frac{1}{T}
			\sum_{k:t_k\leq T}
			B_{q_k},
			\\
			\mathrm{s.t.}
			\quad&
			\mathcal{E}_{1}(t)
			\leq
			\epsilon_1,
			\qquad
			\forall t.
		\end{aligned}
		\label{eq:optimization_problem}
	\end{equation}
	The proposed HTS protocol provides a structured physics-guided policy for this problem by selecting among silence, a macro update, and a micro correction. We do not claim that HTS globally solves
	\eqref{eq:optimization_problem}.


	\section{Physics-Guided Field Triggering and Synchronization}\label{S3}
	Building on the Eulerian field representation, we use the closed-domain divergence--entropy relation as physical motivation for constructing a locally measurable field-deformation indicator. The formal tracking requirement is stated separately through the normalized $L_1$ transport error.

\subsection{Closed-Domain Divergence--Entropy Relation}
\label{sec:entropy_relation}

	For the normalized continuum density
	$p(\mathbf{x},t)=\lambda(\mathbf{x},t)/N$, we use the differential entropy
	\begin{equation}
		H(t)
		\triangleq
		-\int_{\Omega}
		p(\mathbf{x},t)\ln p(\mathbf{x},t)
		\,\mathrm{d}\mathbf{x}.
		\label{eq:differential_entropy}
	\end{equation}
	This quantity describes the dispersion of the reconstructed macroscopic field and should not be interpreted as the entropy of the exact discrete agent micro-state.
	
	For completeness, under the closed-domain condition
	$\mathbf{v}\cdot\mathbf{n}|_{\partial\Omega}=0$, differentiating
	\eqref{eq:differential_entropy} and substituting
	$\partial_t p=-\nabla\cdot(p\mathbf{v})$ gives
	\begin{align}
		\frac{\mathrm{d}H(t)}{\mathrm{d}t}
		&=
		\int_{\Omega}
		(1+\ln p)\nabla\cdot(p\mathbf{v})
		\,\mathrm{d}\mathbf{x} \nonumber\\
		&=
		\int_{\Omega}
		p(\mathbf{x},t)\nabla\cdot\mathbf{v}(\mathbf{x},t)
		\,\mathrm{d}\mathbf{x} \nonumber\\
		&=
		\mathbb{E}_{p}
		[\nabla\cdot\mathbf{v}(\mathbf{x},t)].
		\label{eq:closed_entropy_relation}
	\end{align}
	Equation~\eqref{eq:closed_entropy_relation} is a standard consequence of the continuity equation under the stated no-flux condition. In this work, its role is not to constitute a new entropy identity, but to motivate velocity divergence as a physically interpretable indicator of field deformation. Because positive and negative local divergence can cancel in the spatial expectation, the implemented HTS trigger uses the maximum local absolute divergence together with acceleration, rather than directly using
	$\mathrm{d}H/\mathrm{d}t$.

\subsection{Hierarchical Thermodynamic Synchronization (HTS)}\label{sec:protocol}

	Motivated by the communication--tracking formulation in \eqref{eq:optimization_problem}, we introduce the HTS protocol as a structured physics-guided synchronization policy. By applying Reynolds decomposition, we separate the instantaneous velocity field into a macroscopic mean component and a microscopic fluctuating component: $\mathbf{v}(\mathbf{x},t) = \bar{\mathbf{v}}(t) + \mathbf{v}'(\mathbf{x},t)$. Based on this decomposition, TDT synchronizes the network state across two semantic planes.

\subsubsection{Layer 1: Macro-Semantic Plane}
This layer tracks the mean velocity $\bar{\mathbf{v}}(t)$ that characterizes coherent swarm drift. A macro update is triggered when
$\|\bar{\mathbf{v}}(t)-\hat{\mathbf{v}}_{\mathrm{ref}}\|>\epsilon_{\mathrm{drift}}$,
where $\hat{\mathbf{v}}_{\mathrm{ref}}$ denotes the mean-velocity reference maintained by the digital twin. The edge then sends
$\mathcal{S}_{\mathrm{macro}}=\{\bar{\mathbf{v}}(t)\}$.
The twin updates only the mean component while retaining the previously reconstructed fluctuation field,
\begin{equation}
	\hat{\mathbf{v}}^{+}(\mathbf{x},t)
	=
	\bar{\mathbf{v}}(t)
	+
	\hat{\mathbf{v}}'^{-}(\mathbf{x},t),
\end{equation}
where the superscripts $-$ and $+$ denote the pre- and post-update fields, respectively. The twin then sets
$\hat{\mathbf{v}}_{\mathrm{ref}}=\bar{\mathbf{v}}(t)$.
Because the full spatial field is not reconstructed, a macro update does not reset the most recent micro-correction time $t_r$. The resulting payload remains $\mathcal{O}(1)$.

\subsubsection{Layer 2: Micro-Physical Plane}
This layer handles local deformation represented by the fluctuating component $\mathbf{v}'(\mathbf{x},t)$. A full-state correction is triggered when
\begin{equation}
	\mathcal{D}_{\text{peak}}(t)>\delta_{\text{th}}
	\quad\text{or}\quad
	t-t_r\geq T_{\max},
	\label{eq:micro_trigger}
\end{equation}
where $t_r$ denotes the most recent full-state micro-correction time and $T_{\max}$ is the timeout safeguard derived in Section~\ref{sec:transport_bound}.
The edge then sends
$\mathcal{S}_{\text{micro}}=\{\mathbf{x}_i,\mathbf{v}_i\}_{i\in\mathcal{N}}$,
which reconstructs the density and velocity fields with an $\mathcal{O}(N)$ payload. Thus, the deformation trigger provides early correction, while the timeout prevents indefinitely long extrapolation when the local descriptor remains below threshold. 

\begin{algorithm}[!htbp]
	\caption{Hierarchical Thermodynamic Synchronization}
	\label{alg:hts}
	\small
	\begin{algorithmic}[1]
		\Require Telemetry
		$\{\mathbf{x}_i,\mathbf{v}_i,\mathbf{a}_i\}_{i\in\mathcal N}$;
		$h,\kappa,\gamma,\delta_{\mathrm{th}},
		\epsilon_{\mathrm{drift}},T_{\max}$;
		$\hat{\mathbf v}_{\mathrm{ref}},t_r$
		\Ensure
		$\mathcal S\in
		\{\varnothing,\mathcal S_{\mathrm{macro}},
		\mathcal S_{\mathrm{micro}}\}$
		
		\State
		$\bar{\mathbf v}\gets
		N^{-1}\sum_{i\in\mathcal N}\mathbf v_i$
		
		\For{$i\in\mathcal N$}
		\State
		$\mathcal N_i\gets
		\{j:\|\mathbf x_i-\mathbf x_j\|\leq\kappa h\}$
		\State
		$\lambda_i\gets
		\sum_{j\in\mathcal N_i}W_{ij}$
		\State
		$\phi_i^{\mathrm{div}}\gets
		\left|
		\sum_{j\in\mathcal N_i}
		\frac{1}{\lambda_j}
		(\mathbf v_j-\mathbf v_i)
		\cdot\nabla_iW_{ij}
		\right|$
		\EndFor
		
		\State
		$\mathcal D_{\mathrm{peak}}\gets
		\max_{i\in\mathcal N}
		\{\phi_i^{\mathrm{div}}+\gamma\|\mathbf a_i\|\}$
		
		\If{$\mathcal D_{\mathrm{peak}}>\delta_{\mathrm{th}}$
			\textbf{ or } $t-t_r\geq T_{\max}$}
		\State
		$\hat{\mathbf v}_{\mathrm{ref}}\gets\bar{\mathbf v}$;
		$t_r\gets t$
		\State
		\Return
		$\mathcal S_{\mathrm{micro}}
		=\{\mathbf x_i,\mathbf v_i\}_{i\in\mathcal N}$
		\ElsIf{
			$\|\bar{\mathbf v}-\hat{\mathbf v}_{\mathrm{ref}}\|
			>\epsilon_{\mathrm{drift}}$
		}
		\State
		$\hat{\mathbf v}_{\mathrm{ref}}\gets\bar{\mathbf v}$
		\State
		\Return
		$\mathcal S_{\mathrm{macro}}=\{\bar{\mathbf v}\}$
		\Else
		\State
		Propagate $\hat{\lambda}$ using
		\eqref{eq:digital_extrapolation}
		\State
		\Return $\varnothing$
		\EndIf
	\end{algorithmic}
\end{algorithm}

	The support factor $\kappa$ is determined by the selected compact-support SPH kernel. The timeout $T_{\max}$ limits the interval since the most recent full-state reset, whereas $\mathcal{D}_{\mathrm{peak}}$ provides an earlier locally measurable deformation trigger. Since a micro packet subsumes the mean-velocity information contained in a macro packet, the micro condition is evaluated first, as summarized in Algorithm~\ref{alg:hts}.


	\section{Performance Analysis: Fidelity, Scalability, and Latency}\label{S4}
	We examine TDT in terms of transport-error accumulation, generated synchronization payload, and predictor-assisted closed-loop control.
	
	\subsection{Transport-Error Bound and Silent-Interval Design}
	\label{sec:transport_bound}
	The formal metric $\mathcal{E}_{1}(t)$ and the supplementary diagnostic $\mathcal{E}_{\mathrm{KL}}(t)$ were defined in Section~\ref{sec:problem_formulation}. Only $\mathcal{E}_{1}(t)$ is used in the following transport-error analysis. Let $t_r$ denote the most recent full-state micro-correction time. For the local bound, we consider the silent extrapolation segment immediately following this reset, such that
	$p(\cdot,t_r)=\hat p(\cdot,t_r)$
	and no synchronization update occurs within $(t_r,t]$.
	The physical field evolves under $\mathbf{v}(\mathbf{x},s)$, whereas the digital twin uses the frozen velocity field
	\begin{equation}
		\hat{\mathbf{v}}(\mathbf{x},s)
		=
		\mathbf{v}(\mathbf{x},t_r),
		\qquad s\in[t_r,t].
		\label{eq:frozen_velocity}
	\end{equation}

	\begin{Theorem}
		\label{thm:transport_bound}
		Assume that the physical and extrapolated fields satisfy the same no-flux boundary condition,
		$p(\cdot,t_r)=\hat{p}(\cdot,t_r)$, and
		\begin{equation}
			\|\partial_t\mathbf{v}\|_{\infty}\leq A_{\max},
			\qquad
			\|\nabla\cdot(\partial_t\mathbf{v})\|_{\infty}
			\leq D_{\max}^{\mathrm{acc}},
			\qquad
			A_{\max}>0.
			\label{eq:acceleration_bounds}
		\end{equation}
		Then, for a sufficiently short silent interval,
		\begin{equation}
			\mathcal{E}_{1}(t)
			\leq
			\frac{1}{2}
			C_pA_{\max}(t-t_r)^2
			+
			\mathcal{O}\bigl((t-t_r)^3\bigr),
			\label{eq:transport_bound}
		\end{equation}
		where
		\begin{equation}
			C_p
			\triangleq
			\sup_{s\in[t_r,t]}
			\left(
			\|\nabla p(\cdot,s)\|_1
			+
			\frac{D_{\max}^{\mathrm{acc}}}{A_{\max}}
			\|p(\cdot,s)\|_1
			\right).
			\label{eq:normalized_structural_constant}
		\end{equation}
		Here, $A_{\max}$ bounds the Eulerian temporal variation of the velocity field and should not be confused with the agent acceleration $\mathbf{a}_i$ used in the practical trigger. Since $p$ is normalized, $\|p(\cdot,s)\|_1=1$.
	\end{Theorem}
	\begin{proof}
		Let $e=p-\hat{p}$ and $\Delta\mathbf{v}=\mathbf{v}-\hat{\mathbf{v}}$. Subtracting the physical and extrapolated continuity equations yields
		\begin{equation}
			\partial_t e
			=
			-\nabla\cdot(e\hat{\mathbf{v}})
			-
			\nabla\cdot(p\Delta\mathbf{v}).
		\end{equation}
		Under the common no-flux condition, the transport term does not increase the $L_1$ norm, and hence
		\begin{align}
			\frac{\mathrm{d}}{\mathrm{d}t}\mathcal{E}_{1}(t)
			&\leq
			\int_{\Omega}
			\left|
			\Delta\mathbf{v}\cdot\nabla p
			+
			p\nabla\cdot\Delta\mathbf{v}
			\right|
			\,\mathrm{d}\mathbf{x}.
		\end{align}
		Using
		$\Delta\mathbf{v}
		=
		\partial_t\mathbf{v}|_{t_r}(t-t_r)
		+
		\mathcal{O}((t-t_r)^2)$
		together with \eqref{eq:acceleration_bounds} gives
		\begin{equation}
			\frac{\mathrm{d}}{\mathrm{d}t}\mathcal{E}_{1}(t)
			\leq
			C_pA_{\max}(t-t_r)
			+
			\mathcal{O}((t-t_r)^2).
		\end{equation}
		Integrating from $t_r$ and using $\mathcal{E}_{1}(t_r)=0$ proves \eqref{eq:transport_bound}.
	\end{proof}

Neglecting the higher-order term in \eqref{eq:transport_bound} gives the short-interval design rule
\begin{equation}
	T_{\max}
	\leq
	\sqrt{\frac{2\epsilon_1}{C_pA_{\max}}}.
	\label{eq:tmax_design}
\end{equation}
This timeout is a leading-order safeguard rather than an instantaneous error certificate. In practice, $A_{\max}$ and $D_{\max}^{\mathrm{acc}}$ can be selected as conservative envelopes from offline or recent telemetry, while $C_p$ is evaluated from the reconstructed field over the operating regime. The locally measurable $\mathcal D_{\mathrm{peak}}$ may terminate extrapolation earlier than $T_{\max}$.

\subsection{Regime-Dependent Signaling Scalability}

Let $\rho_{\mathrm{micro}}$ denote the fraction of synchronization decisions requiring a full-state micro-correction, due to either the deformation trigger or the timeout safeguard. Since a micro packet has $\mathcal O(N)$ payload while all non-micro decisions require at most $\mathcal O(1)$ signaling, the average signaling order is
\begin{equation}
	(1-\rho_{\mathrm{micro}})\mathcal O(1)
	+\rho_{\mathrm{micro}}\mathcal O(N).
\end{equation}
Thus, HTS does not remove the worst-case $\mathcal O(N)$ cost; its gain depends on how often full-state corrections are required.

During coherent-motion intervals, HTS therefore uses either silence or a constant-size macro update. When substantial local deformation is detected, or when the timeout expires, the protocol reverts to an $\mathcal O(N)$ full-state micro-correction. Hence, the signaling benefit is amortized rather than worst-case: it depends on the mobility regime through $\rho_{\mathrm{micro}}$. No worst-case $\mathcal O(1)$ signaling claim is made.

\subsection{Predictor-Based Latency Compensation}

Because TDT propagates the cyber state between updates, the same extrapolation mechanism can compensate for a known control-loop delay. For a prediction horizon $\tau_{\mathrm{lead}}$,
\begin{equation}
	\hat{\lambda}(\mathbf{x},t+\tau_{\mathrm{lead}})
	=
	\hat{\lambda}(\mathbf{x},t)
	-
	\int_t^{t+\tau_{\mathrm{lead}}}
	\nabla\cdot(\hat{\lambda}\hat{\mathbf v})\,\mathrm{d}\tau.
\end{equation}
We set $\tau_{\mathrm{lead}}$ according to the nominal round-trip delay. The same frozen-flow velocity field used for field extrapolation is also used to advance the control reference in the closed-loop evaluation. This does not remove the physical delay; its benefit is therefore evaluated empirically.


\section{Numerical Evaluation}\label{S5}

We evaluate TDT through four complementary numerical studies. First, a controlled field-evolution example examines the divergence--entropy relation underlying the proposed deformation cue. We then evaluate computational scalability and the communication--tracking trade-off of the hierarchical synchronization design. Finally, predictor-assisted delay compensation is examined in a closed-loop control setting. Root mean square error (RMSE) is used where coordinate-level tracking accuracy is evaluated.

\subsection{Numerical Check of the Divergence--Entropy Relation}

\begin{figure}[h]
	\vspace*{-3mm}
	\centering
	\includegraphics[width=0.92\columnwidth]{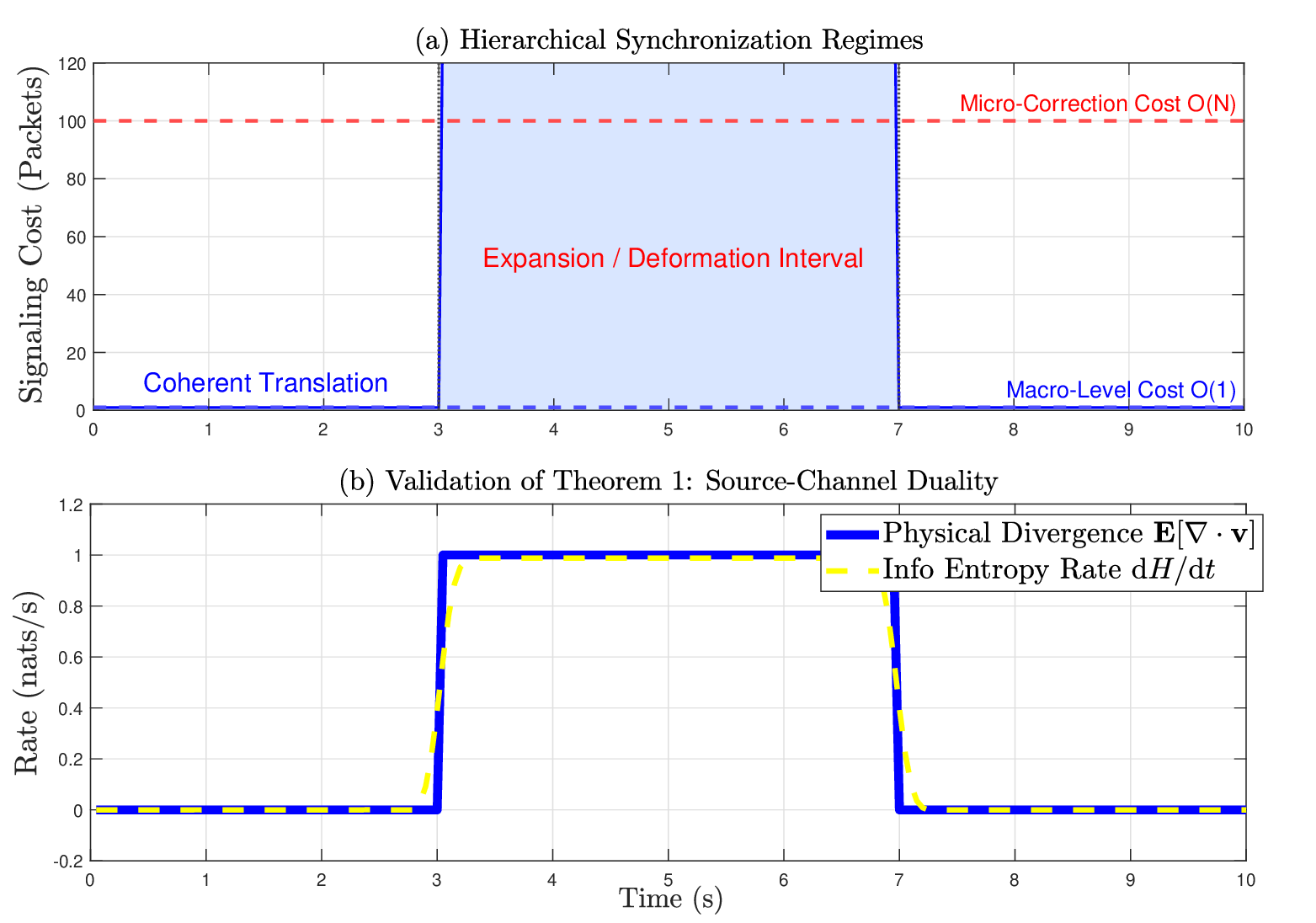}
	\vspace*{-3mm}
	\caption{Numerical consistency between the mean velocity divergence and the entropy-rate curve during a controlled field expansion.}
	\label{fig:theorem_validation}
	\vspace*{-1mm}
\end{figure}
	To examine the physical relation underlying the deformation cue, we consider a controlled transition from coherent translation to radial expansion. The resulting mean-divergence and entropy-rate curves are compared against the closed-domain relation in \eqref{eq:closed_entropy_relation}. As shown in Fig.~\ref{fig:theorem_validation}, the two quantities exhibit consistent temporal evolution during the prescribed expansion interval. This experiment serves as a numerical consistency check of the physical relation rather than a validation of a new mathematical identity.

\subsection{Computational Scalability}

\begin{figure}[h]
	\vspace*{-2mm}
	\centering
	\includegraphics[width=1\columnwidth]{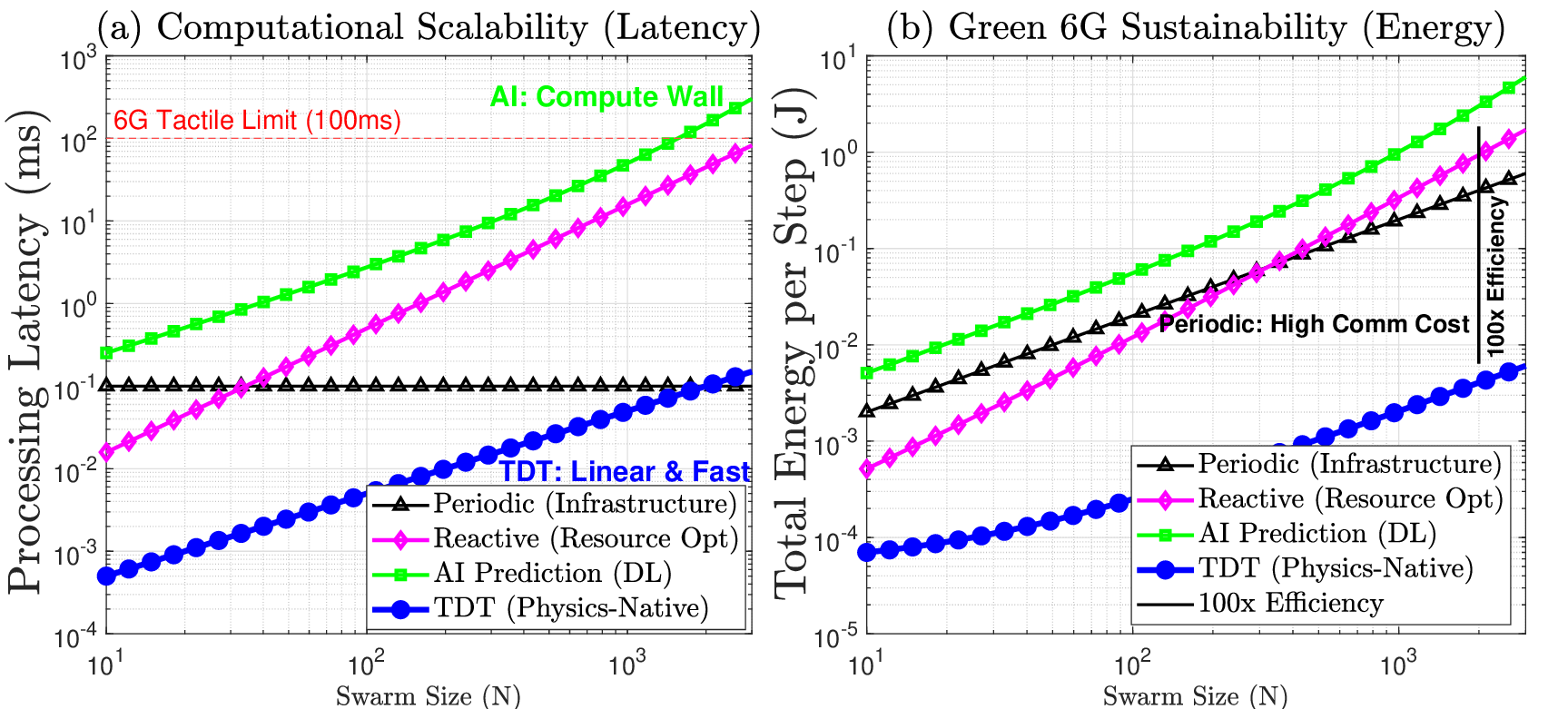}
	\vspace*{-5mm}
\caption{Computational scalability with swarm size: (a) processing latency and (b) per-step energy cost under the adopted cost model.}
	\label{fig:scalability}
	\vspace*{-1mm}
\end{figure}

Realizing the signaling benefit of HTS requires continuous edge-side trigger evaluation. Fig.~\ref{fig:scalability} compares the computational and signaling cost models as the swarm size increases. Under the adopted cost assumptions, TDT exhibits approximately linear processing growth over the evaluated range, whereas the Transformer-based baseline grows more rapidly with $N$. The corresponding energy trend shows a similar scaling advantage, indicating that edge-side trigger evaluation does not dominate the synchronization benefit under the considered model.

\subsection{Communication--Tracking Trade-off}

\begin{figure}[h]
	\vspace*{-3mm}
	\centering
	\includegraphics[width=0.85\columnwidth]{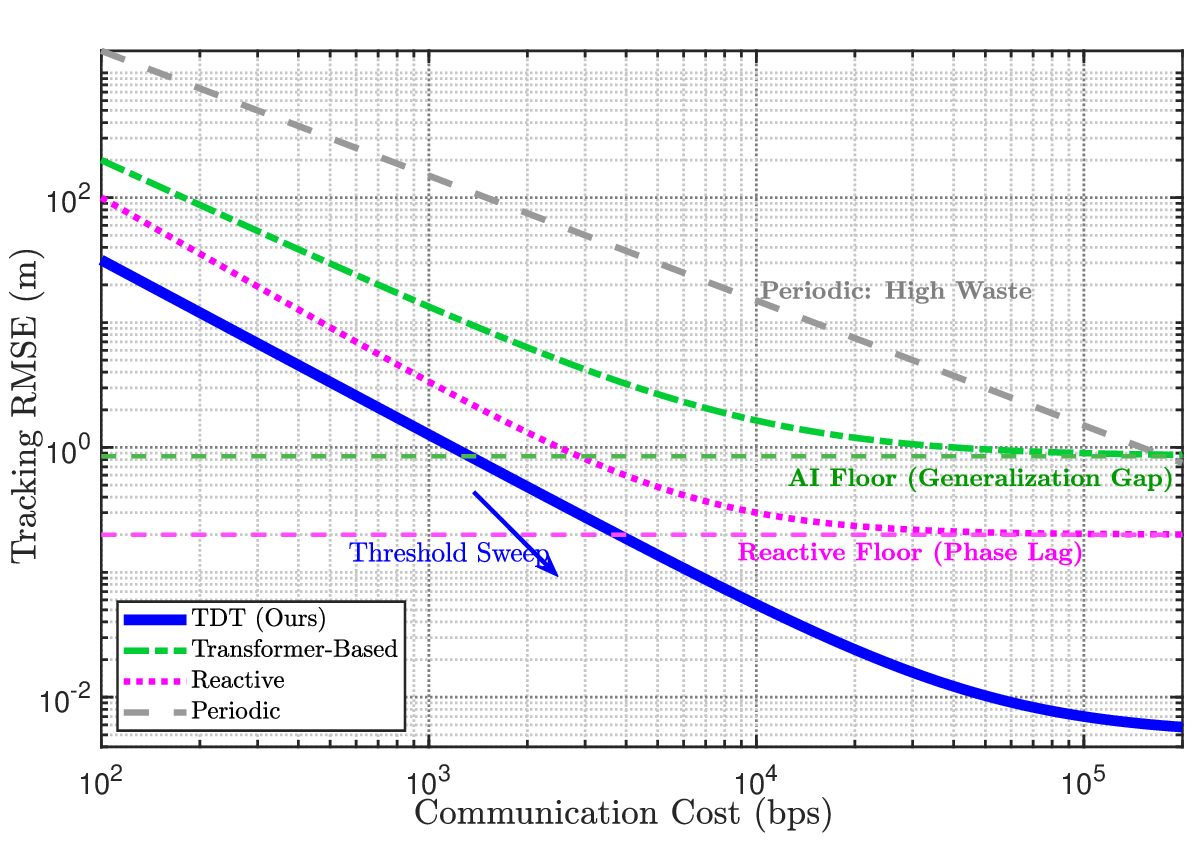}
	\vspace*{-4mm}
	\caption{Communication--tracking trade-off across the evaluated operating points.}
	\label{fig:pareto}
	\vspace*{-2mm}
\end{figure}

Fig.~\ref{fig:pareto} reports the relationship between generated synchronization cost and tracking RMSE across the evaluated operating points. TDT maintains a favorable operating envelope over the considered range, indicating that hierarchical updates can reduce repeated full-state transmission while retaining tracking accuracy. The result is specific to the adopted motion, payload, and triggering models and is not a proof of global Pareto optimality.

\subsection{Predictor-Assisted Closed-Loop Control}
\begin{figure}[h]
	\vspace*{-3mm}
	\centering
	\includegraphics[width=1\columnwidth]{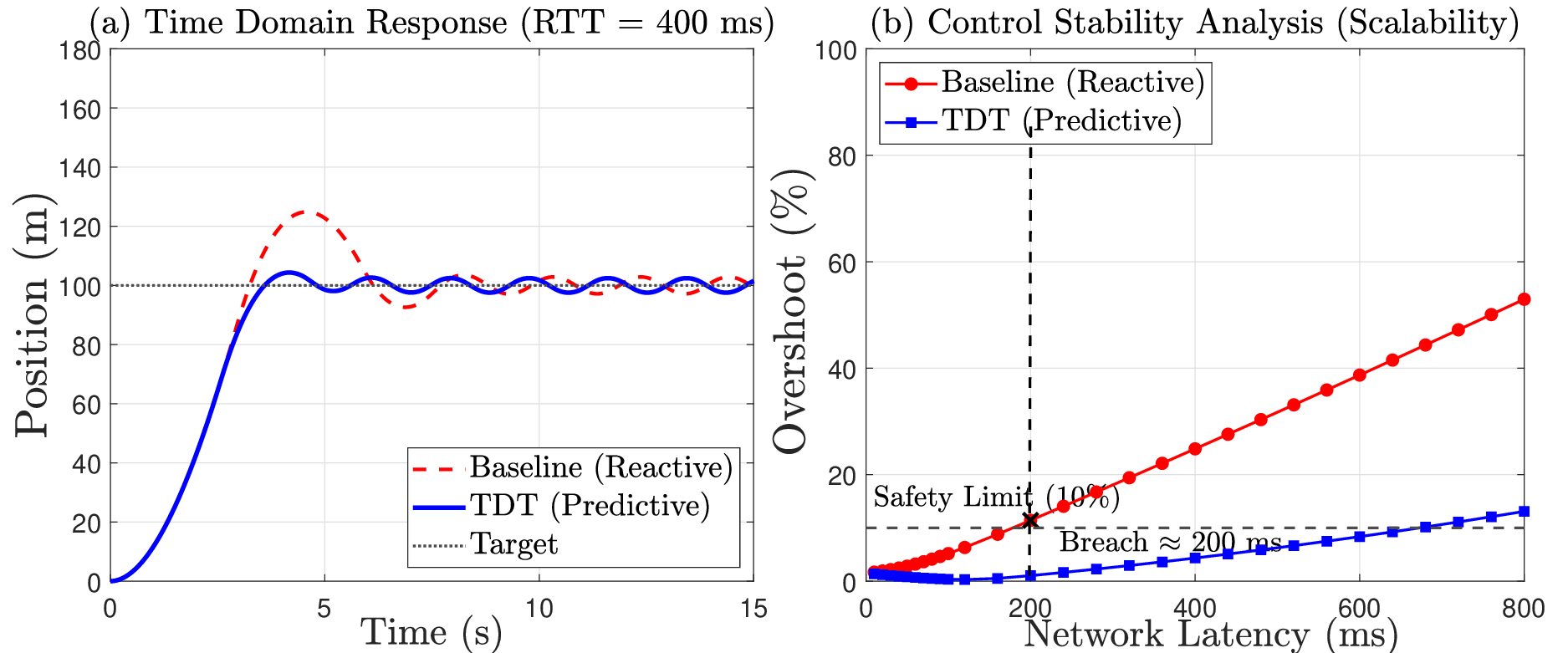}
	\vspace*{-8mm}
	\caption{Closed-loop tracking with predictor-assisted compensation under a 400\,ms round-trip delay.}
	\label{fig:stability_analysis}
	\vspace*{5mm}
\end{figure}

To examine the control-level benefit of the extrapolation mechanism, we use the predicted motion state to advance the reference supplied to a proportional--derivative (PD) controller. We consider a deterministic round-trip delay $\tau_{\mathrm{loop}}=400$\,ms and set $\tau_{\mathrm{lead}}=\tau_{\mathrm{loop}}$, while stochastic jitter is omitted. As shown in Fig.~\ref{fig:stability_analysis}, predictor-assisted control yields lower overshoot and faster settling than the reactive baseline under this delay. A latency sweep further evaluates sensitivity to increasing delay. These results indicate improved delay tolerance under the evaluated controller and delay model, rather than a general stability guarantee under arbitrary stochastic latency.

\section{Conclusion}\label{S6}

This paper presented TDT, a physics-guided synchronization framework for large-scale kinetic-swarm digital twins. The proposed HTS protocol combines low-dimensional macroscopic updates, SPH-based deformation triggering, and full-state microscopic corrections. An $L_1$ transport-error analysis provides a short-interval rule for limiting silent extrapolation, while the local SPH descriptor serves as an implementation-oriented early-warning trigger. Across the evaluated numerical settings, TDT provides a favorable communication--tracking trade-off, while predictor-assisted extrapolation improves delay tolerance in the tested control setting. These results are limited to the adopted closed-domain, motion, telemetry, and payload models; open boundaries, sensing uncertainty, and channel-aware synchronization remain for future work.

\small

\end{document}